\documentclass{article}
\usepackage{amsmath,amsthm,amssymb}
\usepackage{authblk}
\usepackage{fullpage}
\title{Uniting General-Graph and Geometric-Based Radio Networks via Independence Number Parametrization}
\author{Peter Davies}
\affil{Durham University}

\newcommand{\hide}[1]{{}}
\hide{
\usepackage{amsmath}
\usepackage{amssymb}
\usepackage{amsthm}}
\usepackage{bm}
\usepackage{enumerate}
\usepackage[ruled]{algorithm2e}

\newcommand{\Prob}[1]{\mathbb{P}\left[#1\right]}

\newcommand{\MIS}{\ensuremath{\mathbb{MIS}}}
\newcommand{\maxis}{\ensuremath{\alpha}}
\newcommand{\mis}{\ensuremath{\mathfrak{m}}}

\newtheorem{theorem}{Theorem}
\newtheorem{corollary}[theorem]{Corollary}
\newtheorem{lemma}[theorem]{Lemma}

\newtheorem{claim}[theorem]{Claim}

\newcommand{\nat}{\ensuremath{\mathbb N}}

\date{}
\hide{
	\author{Peter Davies}

\authorrunning{Peter Davies} 

\Copyright{Peter Davies} 



\category{} 

\relatedversion{} }

\begin{document}

\maketitle

\begin{abstract}
In the study of radio networks, the tasks of broadcasting (propagating a message throughout the network) and leader election (having the network agree on a node to designate `leader') are two of the most fundamental global problems, and have a long history of work devoted to them. This work has two divergent strands: some works focus on exploiting the geometric properties of wireless networks based in physical space, while others consider general graphs. Algorithmic results in each of these avenues have often used quite different techniques, and produced bounds using incomparable parametrizations.

In this work, we unite the study of general-graph and geometric-based radio networks, by adapting the broadcast and leader election algorithm of Czumaj and Davies (JACM '21) to achieve a running-time parametrized by the {\it independence number} of the network (i.e., the size of the maximum independent set). This parametrization preserves the running time on general graphs, matching the best known, but also improves running times to near-optimality across a wide range of geometric-based graph classes. 

As part of this algorithm, we also provide the first algorithm for computing a maximal independent set in general-graph radio networks. This algorithm runs in $O(\log^3 n)$ time-steps, only a $\log n$ factor away from the $\Omega(\log^2 n)$ lower bound.
\end{abstract}

\section{Introduction}
The radio network model is an abstraction of networks of wireless devices, which aims to capture the essential behavior of wireless transmissions and interference, while simplifying the physical details. 

\subsection{Model}
The network is represented by an undirected graph $G=(V,E)$, where the nodes represent transmitter-receiver devices, and the edges represent direct reachability with wireless transmission between pairs of devices. We denote by $n$ the number of nodes in the graph, and by $D$ the diameter (the maximum distance between a pair of nodes in the graph). We will also use the independence number $\maxis$, which is the size of the maximum independent set in the graph.
Time is divided into synchronous time-steps, with each node making a choice to either transmit a message or listen for messages in each time-step. The defining behavior of radio networks is that listening nodes hear a message only if \emph{exactly} one of their neighbors transmits; otherwise, the wireless signals are assumed to have interfered with each other, providing the listening node with no useful information\footnote{This abstraction is in some sense worst-case. On the geometric-based side of radio network research, the signal-to-interference-plus-noise ratio (SINR) model uses an alternative criterion for hearing transmissions based more closely on the physical behavior of wireless transmissions; see, e.g., \cite{DGKN13}.}.

We work in the version of the model that is:

\begin{itemize}
	\item Ad-hoc: that is, nodes have no prior knowledge of the graph structure, their own degree, or their neighbors. We do assume that nodes have knowledge of (at least linear upper estimates of) the graph parameters $n$ and $D$. Our algorithm also uses the independence number of the network $\maxis$, but here any polynomial approximation will suffice; for our results on growth-bounded graphs, for example, it would suffice for nodes to know that $\maxis = poly(D)$.
	
	\item Without collision detection: that is, listening nodes cannot distinguish between hearing multiple transmission, and hearing no transmission.
	
	\item With synchronous wake-up: all nodes `wake up' in the same time-step and can begin operation immediately.
\end{itemize}

Since we study randomized algorithms in this work, we may assume that nodes are initially indistinguishable, but have access to their own private sources of randomness. Then, as a linear upper estimate of $n$ is known, nodes can uniformly randomly generate identifiers (IDs) from e.g. $[O(n^3)]$, which will be unique with high probability.

\subsection{Problems}
We study three fundamental problems in this work, all of which are well established in radio networks and related models:

\paragraph{Broadcasting} Broadcasting is the most basic global task: a single designated source node holds a message, and a broadcasting algorithm must ensure that upon completion, all nodes in the graph are aware of the message. We assume that the graph is connected, since otherwise broadcasting is impossible.

\paragraph{Leader Election} Leader Election is a self-organisation task often used as the starting point for more complex procedures. It involves all nodes in the network agreeing on the ID of one particular (arbitrary) node in the network to designate as leader. Again, we must assume the graph is connected for this to be possible.

\paragraph{Maximal Independent Set} Maximal Independent Set (MIS) entails all nodes deciding whether or not be part of an output set $\MIS\subseteq V$, which must satisfy two properties:

\begin{itemize}
	\item Independence: there can be no edge between two nodes in \MIS;
	\item Maximality: no vertex could have been added to \MIS\ without violating independence, or equivalently, every node $v\notin \MIS$\ has a neighbor in \MIS.
\end{itemize}

Unlike the previous two problems, MIS is a `local' problem, in that it can be solved and verified by nodes only communicating within a radius much smaller than the diameter of the graph. Consequently, we need not assume that the graph is connected when considering the MIS problem.

\subsection{Graph classes}\label{sec:classes}
Our algorithms will work for general undirected graphs, and do not require or use any link between graph nodes and geometric positions. However, the broadcasting and leader election algorithms will give improved results for some special families of graphs, which incorporate those arising from geometric wireless communications networks:

\paragraph{Growth-bounded graphs} (Polynomially) growth-bounded graphs are those in which, for all nodes $v$ and radii $d\in \nat$, any independent set in the $d$-hop neighborhood of $v$ has at most $poly(d)$ size. In particular, this means that any independent set in the entire graph has at most $poly(D)$ size, so growth-bounded graphs are a subset of those that have $\maxis = poly(D)$, the property that we use to obtain near-optimal running times.

While the definition of growth-bounded graphs is not geometric, it captures behavior in geometric-derived classes for the following reason: if adjacency in a graph class is based on some concept of geometric \emph{closeness}, then an independent set represents a \emph{packing} of nodes that are too \emph{distant} to be adjacent. The number of nodes that can be packed within a certain radius is then bounded by some function of that radius (and, for e.g. constant-dimensional Euclidean space, a \emph{polynomial} function).

For this reason, the class of growth-bounded graphs includes, to our knowledge, all of the geometric-inspired classes that have previously been studied to capture radio networks, including the following:

\paragraph{Unit disk graphs} Unit disk graphs are graphs in which the nodes have positions in two-dimensional Euclidean space, and edges are placed between nodes which are within distance $1$ of each other. Unit disk graphs are growth-bounded, with any independent set in the $r$-hop neighborhood of $v$ having $O(r^2)$ size.

\paragraph{Quasi unit disk graphs} Quasi unit disk graphs are a generalization of unit disk graphs, which relaxes the edge condition up to some constant factor while still requiring the graphs to be undirected. Specifically, if the Euclidean distance between two nodes is less than some $r$, then there must be an edge between them, if it is larger than some $R > r$ then there must not, and if it is between $r$ and $R$ then there may or may not be an edge. Quasi unit disk graphs are also growth-bounded: all nodes within (graph) distance $d$ of some node $v$ must be within Euclidean distance $dR$, and since any two independent nodes must be at least distance $r$ apart, there can be at most $O\left((\frac{dR}{r})^2\right)$ within a $dR$-radius of $v$. The ratio $R/r$ is generally treated as a constant, in which case this is $O(d^2)$.

\paragraph{Unit ball graphs} Unit ball graphs extend unit disk graphs to allow the underlying metric space to be any metric space. One can also define quasi unit ball graphs analogously. These graphs are again growth-bounded if the underlying metric space is \emph{doubling}. (A metric space $X$ with metric $dist$ is said to be doubling if there is some doubling constant $b \in \nat$ such that for any $x\in X$ and $r > 0$, the ball $B(x, r) := \{y | dist(x, y) < r\}$ can be covered by at most $b$ balls of radius $r/2$.) Then, for a (quasi) unit ball graph arising from such a metric space (with doubling constant $b$), the size of an independent set in the $d$-hop neighborhood of any $v$ is $d^{O(b)}$. The most obvious application of this is that (quasi) unit ball graphs in any fixed-dimensional Euclidean space are growth-bounded.

\paragraph{Geometric radio networks} Geometric radio networks are a different generalization of unit disk graphs, in which, in addition to spacial positions, each node $v$ also has a \emph{range} value $r_v$. A directed edge is then drawn from $v$ to some node $u$ if the Euclidean distance between $u$ and $v$ is at most $r_v$. Geometric radio networks are again growth bounded, so long as the ratio between the largest and smallest range values is constant. While similar to quasi unit disk graphs, the major difference is that unlike the above graph classes, geometric radio networks produce directed graphs. Since the scope of this work is limited to undirected graphs, we restrict our consideration of this class to the subclass of geometric radio networks that are also undirected.

\subsection{Notation}
We use the phrase `with high probability' to refer to events that occur with probability at least $1-n^{-c}$, for some constant $c\ge 1$. \footnote{The value of the constant is not important, since all of our uses will refer to properties of randomized algorithms, and in all cases we can amplify $c$ to be arbitrarily high by increasing the running-time of the process by a constant factor.} We use $N(v)$ to denote the set of neighbors of a node $v$, i.e. $\{u\in V: \{u,v\}\in E\}$.
\subsection{Related Work}
We first survey work on general graphs, and then on geometric-derived graph classes. 

\subsubsection{General graphs}
\paragraph{Broadcasting}
The study of randomized broadcasting began with the seminal \textsc{Decay} algorithm of Bar-Yehuda, Goldreich and Itai \cite{BGI92}, achieving an $O(D\log n +\log^2 n)$ time-step upper bound. Subsequently, Czumaj and Rytter \cite{CR06} and Kowalski and Pelc \cite{KP03} independently gave randomized algorithms taking $O(D\log\frac nD +\log^2 n)$ time-steps, matching the known $\Omega(D\log\frac nD +\log^2 n)$ lower bounds \cite{ABLP91,KM98}. However, the lower bound of Kushilevitz and Mansour \cite{KM98} assumed that spontaneous transmissions (nodes transmitting before receiving the broadcast message) were not allowed. Haeupler and Wajc \cite{HW16} exploited the capability for spontaneous transmissions to obtain an $O(D\log_D n\log\log n+\log^{O(1)}n)$ randomized broadcast algorithm for undirected networks, surpassing this lower bound. Czumaj and Davies \cite{CD21} improved this bound to $O(D\log_D n+\log^{O(1)}n)$, and also extended the algorithm to perform leader election in the same running time.

Deterministic algorithms have also been studied: in undirected networks, the fastest known algorithm is an $O(n \log D)$-time one due to Kowalski \cite{K05}, while in directed networks the $O(n\log D \log\log D)$ running time of Czumaj and Davies \cite{CD18} is the best known.

\paragraph{Leader Election}
It is well-known that leader election can be completed with high probability in $O(\log n) \times$ \emph{broadcasting time} using a binary search for the highest ID (provided that the broadcasting algorithm extends to multiple sources). Similarly, having nodes choose to become a candidate leader with $\Theta(1/n)$ probability and then performing a (multiple-source) broadcast to check whether a single leader has been elected succeeds in $O(\text{\emph{broadcasting time}})$ rounds \emph{in expectation} \cite{CD19}. Prior to the result of \cite{CD21}, Ghaffari and Haeupler \cite{GH13} achieved a running time of $O(D\log\frac{n}{D}+\log^3 n) \cdot \min\{\log\log n,\log\frac nD\}$ (for high-probability success), which remains the fastest known for a range of $D$ between $\log^2 n\log\log n$ and $\log^{O(1)} n$. 

For deterministic leader election, in undirected radio networks an $O(n\log^{3/2}n\sqrt{\log\log n})$ algorithm is known \cite{CKP12}, and an $O(n\log n\log D\log\log D)$ running time can be obtained (also in directed graphs) by applying the above binary search approach with the (multi-source) broadcasting algorithm of \cite{CD18}.

\paragraph{Maximal Independent Set}
Maximal independent set is considered one of the most fundamental symmetry-breaking tasks in distributed computing, and has seen substantial study in message-passing models (see e.g. \cite{Luby86,Ghaffari16}). In single-hop networks (those in which the underlying graph is a clique), it is equivalent to leader election, since a single successful transmission suffices for both tasks. Surprisingly, however, we do not know of any prior maximal independent set algorithm for radio networks on general graphs.

Lower bounds do exist, via a reduction to the single-hop wake-up problem (see, e.g., \cite{MW05}). This problem is as follows: $n$ nodes are arranged in a clique, but only some unknown $k\in [1,n]$ of them are \emph{activated} in time-step 0. The goal is to guarantee a successful transmission from one of these nodes (i.e., a time-step in which exactly one of the active nodes transmits, and all others do not). Any MIS algorithm which succeeds with $1-o(1)$ probability, simulated by the active nodes, must ensure such a successful transmission\footnote{It may seem that the MIS algorithm may fail because it is given $n$ as the network size, but only run on $k$ nodes. However, a correct MIS algorithm must still succeed in such a situation, since the perspective of the clique nodes is identical to that if there were $n-k$ extra isolated nodes in the network, in which case the network size parameter $n$ would be correct.}. Therefore, the $\Omega(\log^2 n)$ lower bound (for success with high probability) of Farach-Colton, Fernandes, and Mosteiro \cite{FFM06} applies.

\subsubsection{Geometric-derived graph classes}

\paragraph{Broadcasting and leader election}
Schneider and Wattenhofer \cite{SW10} give an $O(D \log n)$-round deterministic broadcasting algorithm for growth-bounded graphs \emph{with collision detection}, an $\Omega(D+\log n)$-round randomized lower bound with collision detection, and an $\Omega(n \log_{n/D} D)$ deterministic lower bound without collision detection.

Onus, Richa, Kothapalli, and Scheideler \cite{ORKS05} give an algorithm for quasi-unit disk graphs which requires $O(\Delta \log \Delta \log n + \log^4 n)$ rounds (where $\Delta$ is the maximum degree in the graph) to set up a constant-density spanner structure (using a construction from \cite{KSOR05}) and then $O(D+\log n)$ rounds (randomized) for subsequent broadcasts. However, unless $\Delta$ is significantly lower than $D$, the spanner computation exceeds the time for a single broadcast in general graphs, so is only efficient for multiple broadcasts.

Emek, Gasieniec, Kantor, Pelc, Peleg, and Su \cite{EGKPS07} study unit disk graphs, and give bounds based on their granularity $g$ (the inverse of the minimum Euclidean distance between two nodes) of $\Theta(\min\{D+g^2, D\log g\})$ for deterministic algorithms. To compare between parametrizations, note that $g=\Omega(\frac{\sqrt n}{D})$, by an area argument. So, the running time bound of \cite{EGKPS07} parameterized only by $n$ and $D$ is at least $\Omega(\min\{D+\frac{n}{D^2}, D\log n\})$ (but can be arbitrarily high).

Dessmark and Pelc \cite{DP07} work on geometric radio networks. They give an $O(D+\log n)$-round deterministic algorithm for \emph{undirected} geometric radio networks, but this algorithm requires nodes to know their spatial coordinates and have access either to collision detection or to extra initial knowledge of their neighborhoods, so does not apply to the standard ad-hoc model.

The above works all study broadcasting; we are not aware of any dedicated leader election algorithm specifically for any of the above graph classes, though it is likely that some of the above broadcasting algorithms could be adapted for the purpose.

\paragraph{Maximal Independent Set}
Algorithms for maximal independent set in geometric graphs classes are known: Moscibroda and Wattenhofer provided such an algorithm for unit disk graphs with a running time of $O(\log^9 n/\log\log n)$ \cite{MW04}, which they then improved to an optimal $O(\log^2 n)$ \cite{MW05} (the $\Omega(\log^2 n)$ lower bound mentioned above \cite{FFM06} still applies, since it is proven on a clique, which is a UDG). Daum, Ghaffari, Gilbert, Kuhn, and Newport \cite{DGGKN13} gave an algorithm for \emph{multi-channel} radio networks (which generalize the \emph{single-channel} networks we study by having multiple simultaneous communication channels, in which only transmissions on the same channel collide). Translated to single-channel networks, their running time is $O(\log^2 n)$ on bounded-independence graphs. Daum and Kuhn \cite{DK15} then both improved the running time for multi-channel networks (though it remains $O(\log^2 n)$ for single-channel networks), and extended the graph family from \emph{polynomially} growth-bounded graphs to graphs where independent sets within a radius of $d$ are bounded by \emph{any} function of $d$ (that is, for any fixed constant $d$, the size of an independent set within a $d$-radius neighborhood is bounded by a constant, but when $d$ grows with $n$ there is no nontrivial bound).

\subsection{Our Results} We extend the algorithm of \cite{CD21} to perform both broadcast and leader election in $O(D\log_D \maxis+\log^{O(1)} n)$ rounds, where $\maxis$ is the \emph{independence number} of the graph (i.e., the size of the maximum independent set). This maintains the $O(D\log_D n+\log^{O(1)} n)$ bound of \cite{CD21} in general graphs, since $\maxis\le n$, but also takes only $O(D+\log^{O(1)} n)$ rounds in graphs with $\maxis = D^{O(1)}$, which includes all growth-bounded graphs. Since $\Omega(D)$ rounds are clearly required for both broadcasting and leader election, the leading term here is optimal.

In doing so, we improve over previous results in geometric-derived graph classes, without requiring any geometric information or representation. While such prior results are often not directly comparable, they are all either slower for most parameter ranges \cite{SW10,EGKPS07,ORKS05} or require collision detection \cite{SW10,DP07} or other additional capabilities such as geometric coordinates \cite{EGKPS07,DP07}.

As part of this broadcasting and leader election algorithm, we also provide the first algorithm for maximal independent set in general-graph radio networks, achieving a running time of $O(\log^3 n)$. This comes within a $\log n$ factor of the lower bound \cite{FFM06}.

\section{Our Approach}
We achieve our new bound for broadcasting and leader election as follows: we first compute a maximal independent set, using Algorithm \ref{alg:RMIS} (see Section \ref{sec:MIS}), a radio network adaptation of Ghaffari's algorithm for MIS in the LOCAL model \cite{Ghaffari16}. The resulting MIS will serve as a set of potential cluster centers in an adaptation of the clustering algorithm of Miller, Peng and Xu \cite{MPX13}. The broadcasting and leader election algorithms then follow a similar line as those of \cite{CD21}. The main technical difference is an improvement to a crucial part of the analysis of Miller, Peng and Xu's clustering, which provides a bound on the expected distance from nodes to their cluster centers based on the number of nearby MIS nodes. This improved bound is not specific to the radio network model and may prove of independent interest.

\subsection{Discussion of the Algorithm of \cite{CD21}}

The main procedure of \cite{CD21}, which can be used to perform both broadcasting and leader election, is called \textsc{Compete$(S)$} and works as follows (Algorithm \ref{alg:C}): the procedure begins with a set $S$ of `candidates' who all hold messages, and aims to have them propagate their messages throughout the network. When two messages come into contact, the one that is lexicographically `higher' (for broadcasting and leader election, the order is not important so long as it is consistent) will override the other. By the end of the procedure, all nodes will therefore know the highest message.

	\begin{algorithm}[H]
	\caption{\textsc{Compete$(S)$ - Main Process}}
	\label{alg:C}
	\begin{enumerate}[1)]
		\item Compute a \emph{coarse clustering} using $\textsc{Partition}(\beta)$ with $\beta = D^{-0.5}$.
		\item Compute a schedule within each coarse cluster.
		\item Within each coarse cluster, for each integer $j$ with $0.01 \log D \le j \le 0.1 \log D$, compute $D^{0.2}$ different \emph{fine clusterings} using $\textsc{Partition}(\beta)$ with $\beta = 2^{-j}$.
		\item Compute schedules within all fine clusterings.
		\item Each coarse cluster center computes a $D^{0.99}$-length sequence of randomly chosen fine clusterings to use.
		\item Transmit this sequence within each coarse cluster, using the coarse cluster schedules.
		\item For each fine clustering in the sequence perform \textsc{Intra-Cluster Propagation$(O(\frac{\log n}{\beta\log D}))$} (with the value of $\beta$ corresponding to the fine clustering), terminating after $O(D\frac{\log n}{\log D})$ rounds.
	\end{enumerate}
\end{algorithm}

The algorithm heavily utilizes the $\textsc{Partition}(\beta)$ process to create clusterings. This process is a radio network implementation (due to \cite{HW16}) of Miller, Peng and Xu's graph clustering \cite{MPX13}. If applied with a fixed parameter $\beta$, it creates a clustering with cluster diameters $O(\frac{\log n}{\beta})$ with high probability. However, Haeupler and Wajc \cite{HW16} showed that if $\beta$ is chosen randomly from some range polynomial in $D$, then in expectation the distance from a node to its cluster center is only $O(\frac{\log_D n\log\log n}{\beta})$, a bound that Czumaj and Davies \cite{CD21} then improved to $O(\frac{\log_D n}{\beta})$:

	\begin{theorem}[\cite{CD21} Theorem 2.2]
	\label{thm:oldcprop}
	Let $j$ be an integer chosen uniformly at random between $0.01\log D$ and $0.1\log D$, and let $\beta = 2^{-j}$. For any node $v$, with probability at least $0.55$ (over choice of $j$), the expected distance from $v$ to its cluster center upon applying $\textsc{Partition}(\beta)$ is $O(\frac{\log_D n}{\beta})$.
\end{theorem}

This is exploited by creating a collection of \emph{fine clusterings}, with values of $\beta$ within this range, and cycling between them at random. Communication throughout the network is then done by performing \textsc{Intra-Cluster Propagation}, which facilitates fast communication within the \emph{fine} clusters, using the schedules from \cite{GHK13} as implemented in \cite{HW16} (Algorithm \ref{alg:ICP}).

However, to construct and use these fine clusterings, the random choices of $\beta$ must be coordinated between nodes, and to do this, the algorithm first employs a \emph{coarse} clustering. The coarse clustering uses the same $\textsc{Partition}(\beta)$ process, but without randomizing $\beta$: this means that it is too slow to use for global communication (broadcast and leader election using it would take $O(D \log n + \log^{O(1)} n)$ rounds), but suffices to locally agree values of $\beta$ for the fine clusterings.

There are other components to the algorithm of \cite{CD21}: a background process, run concurrently with Algorithm \ref{alg:C} via time multiplexing, for passing messages over coarse cluster boundaries, and the details of the \textsc{Intra-Cluster Propagation} procedure. However, the only changes we make in order to achieve the new dependence on independence number are to \textsc{Compete$(S)$} and $\textsc{Partition}(\beta)$, so we need not be too concerned with the workings of these other components. We give a brief summary in Appendix \ref{app:components}; for further detail, see \cite{CD21}.

\subsection{Changes to the Algorithm of \cite{CD21}}\label{sec:changes}

We make one major change to the algorithm of \cite{CD21}, which concerns how the clusterings are created. In \cite{CD21}, the $\textsc{Partition}(\beta)$ process of \cite{HW16} is used, which is a radio-network implementation of the following clustering process of Miller, Peng, and Xu \cite{MPX13}:

\begin{itemize}
	\item Each node $v$ independently chooses a random variable $\delta_v$ from an exponential distribution with parameter $\beta$.
	\item Each node $u$ joins the cluster of the node $v$ minimizing $dist(u,v)-\delta_v$.
\end{itemize}

To exploit the properties of growth-bounded graphs, we make the following change: we first compute a maximal independent set \MIS\ using our new procedure \textsc{ComputeMIS} (Algorithm \ref{alg:RMIS}, from Section \ref{sec:MIS}). Then, we replace the standard $\textsc{Partition}(\beta)$ with a variant $\textsc{Partition}(\beta,\MIS)$ that uses \emph{only \MIS\ nodes} as potential cluster centers, i.e.:

\begin{itemize}
	\item Each \MIS\ node $v$ independently chooses a random variable $\delta_v$ from an exponential distribution with parameter $\beta$.
	\item Each node $u$ joins the cluster of the \MIS\ node $v$ minimizing $dist(u,v)-\delta_v$.
\end{itemize}

It can easily be seen that this change does not materially affect the radio network implementation of \cite{HW16}, and the clusterings of $\textsc{Partition}(\beta,\MIS)$ can still be computed in $O(\frac{\log^{O(1)}n}{\beta})$ rounds.

Our updated variant of the main \textsc{Compete} procedure of \cite{CD21} is then as follows:

\begin{algorithm}[H]
	\caption{\textsc{Compete$(S)$}}
	\label{alg:CC}
	\begin{enumerate}[1)]
		\item $\MIS \gets \textsc{ComputeMIS}$.		
		\item Compute a \emph{coarse clustering} using $\textsc{Partition}(\beta,\MIS)$ with $\beta = D^{-0.5}$.
		\item Compute a schedule within each coarse cluster.
		
		\item Within each coarse cluster, for each integer $j$ with $0.01 \log D \le j \le 0.1 \log D$, compute $D^{0.2}$ different \emph{fine clusterings} using $\textsc{Partition}(\beta,\MIS)$ with $\beta = 2^{-j}$.
		\item Compute schedules within all fine clusterings.
		\item Each coarse cluster center computes a $D^{0.99}$-length sequence of randomly chosen fine clusterings to use.
		\item Transmit this sequence within each coarse cluster, using the coarse cluster schedules.
		\item For each fine clustering in the sequence perform \textsc{Intra-Cluster Propagation$(O(\frac{\log_D \maxis}{\beta}))$} (with the value of $\beta$ corresponding to the fine clustering), terminating after $O(D\log_D \maxis)$ rounds.
	\end{enumerate}
\end{algorithm}

The result of this change is that we can improve over Theorem 2.2 of \cite{CD21}: when clusterings are generated using $\textsc{Partition}(\beta,\MIS)$, we can show the following theorem:

\begin{theorem}
	\label{thm:cprop}
	Let $j$ be an integer chosen uniformly at random between $0.01\log D$ and $0.1\log D$, and let $\beta = 2^{-j}$. For any node $v$, with probability at least $0.77$ (over choice of $j$), the expected distance from $v$ to its cluster center upon applying $\textsc{Partition}(\beta)$ is $O(\frac{\log_D \maxis}{\beta})$.
\end{theorem}

The other change in Algorithm \ref{alg:CC} is that we have shortened the length of time that \textsc{Intra-Cluster Propagation} is run to match the new analysis.

\section{Analysis}
In this section, we discuss the changes, simplifications, and extensions made to the analysis of \cite{CD21} in order to achieve our new running-time result. As there, we use quantities $T_{\beta}$, $B_{\beta}$ and $S_{\beta}$ to bound the expected distance from any fixed $v$ to its cluster center after applying $\textsc{Partition}(\beta,\MIS)$. However, their definitions are slightly altered to account for the modified clustering procedure. For fixed $v$ and computed maximal independent set \MIS, let $\mis_i$ denote the number of \MIS\ nodes at distance $i$ from $v$. Then, denote $T_{\beta} := \sum_{i=0}^{D} i \mis_i e^{-i\beta}$, $B_{\beta} := \sum_{i=0}^{D} \mis_i e^{-i\beta}$, and $S_{\beta} := \frac{T_{\beta} }{B_{\beta}} = \frac{\sum_{i=0}^{D} i \mis_i e^{-i\beta}}{\sum_{i=0}^{D} \mis_i e^{-i\beta}}$.

We start with a lemma showing that the expected distance from a node to its cluster center can be bounded in terms of $S_{\beta}$:
	
\begin{lemma}
	\label{lem:edist}
	For any fixed node $v$ and value $\beta\le D^{-0.01}$, the expected distance from $v$ to its cluster center upon applying $\textsc{Partition}(\beta,\MIS)$ is at most $\frac{5 \sum_{i=0}^{D} i \mis_i e^{-i\beta}}{\sum_{i=0}^{D} \mis_i e^{-i\beta}} = 5S_{\beta}$.
\end{lemma}

The proof is identical to that of Lemma 6.1 of \cite{CD21}, other than the notational change from $x_i$ to $\mis_i$; this change does not affect the proof itself. The main technical change is Lemma \ref{lem:cl:goodj}, which is both stronger and has a substantially simpler proof than the equivalent result from \cite{CD21} (for which the proof runs to $4$ pages).

Let $b = 2^{\lceil\log_2\log_D \maxis \rceil+2}$; notice that $b$ is an integer power of $2$ with $2\le 4\log_D \maxis\le b \le 8\log_D \maxis$. Furthermore, for $j \in \nat$, let $s_{j}$ denote $\sum\limits_{i=0}^{2^{j+1}}\mis_i$.

The next lemma essentially says that if the number of $\MIS$ nodes does not expand too quickly just outside the radius $\frac{\log_D \maxis}{\beta}$ around $v$, then we indeed have the desired $O(\frac{\log_D \maxis}{\beta})$ bound on the expected distance from $v$ to its cluster center under $\textsc{Partition}(\beta,\MIS)$. 

\begin{lemma}\label{lem:cl:goodj}
	If for some $j\in \nat$ and for  all $r\ge 8$, $s_{j+\log b+r}\le 2^{b2^{r-1}}s_{j+\log b}$, then when $\beta=2^{-j}$, $S_{\beta} = O(b2^j) = O(\frac{\log_D \maxis}{\beta})$
\end{lemma}

\begin{proof}
	Let $q' :=  \log b + 7$. We bound the numerator $T_{2^{-j}}$ of $S_{2^{-j}}$:
	\begin{align*}
	T_{2^{-j}}
	&=
	\sum_{i=0}^D i \mis_i e^{-i2^{-j}}\\
	&=
	\sum_{i=0}^{2^{j+q'}} i \mis_i e^{-i2^{-j}} + \sum_{i=2^{j+q'}+1}^{D} i \mis_i e^{-i2^{-j}}\\
	&\le
	2^{j+q'}\sum_{i=0}^{2^{j+q'}} \mis_i e^{-i2^{-j}} + \sum_{q=q'}^{\infty}\sum_{i=2^{j+q}+1}^{2^{j+q+1}} i \mis_i e^{-i2^{-j}}\\
	&\le b2^{j+7}B_{2^{-j}} + \sum_{q=q'}^{\infty}2^{j+q+1}\sum_{i=2^{j+q}+1}^{2^{j+q+1}}  \mis_i e^{-2^{j+q} \cdot2^{-j}}\\
	&\le b2^{j+7}B_{2^{-j}} + \sum_{q=q'}^{\infty}2^{j+q+1}\sum_{i=0}^{2^{j+q+1}}  \mis_i e^{-2^{q}}\\
	&= b2^{j+7}B_{2^{-j}} + 2^{j+1}\sum_{q=q'}^{\infty}2^{q}e^{-2^{q}}s_{j+q+1}
	\enspace.
	\end{align*}
	
	We apply the condition of the lemma to upper-bound the terms $s_{j+q+1}$ by $2^{b2^{q-\log b}}s_{j+\log b} = 2^{2^{q }}s_{j+\log b}$:
	
	\begin{align*}
	T_{2^{-j}}
	&\le b2^{j+7}B_{2^{-j}} + 2^{j+1}\sum_{q=q'}^{\infty}2^{q}e^{-2^{q}}2^{2^{q }}s_{j+\log b}\enspace.
	\end{align*}
	
	Next, we bound $s_{j+\log b}$ in terms of $B_{2^{-j}}$:
	
	\begin{align*}
	s_{j+\log b} &= e^b\sum\limits_{i=0}^{2^{j+\log b}}\mis_i e^{-b} = 
	e^b\sum\limits_{i=0}^{2^{j+\log b}}\mis_i e^{-b2^j\cdot 2^{-j}}
	\le e^b\sum\limits_{i=0}^{2^{j+\log b}}\mis_i e^{-i\cdot 2^{-j}}
	\le e^b \sum\limits_{i=0}^{d}\mis_i e^{-i 2^{-j}}
	= e^b B_{2^{-j}}
	\enspace.
	\end{align*}
	
	So plugging this into our above bound for $T_{2^{-j}}$,
	\begin{align*}
	T_{2^{-j}}
	&\le b2^{j+7}B_{2^{-j}} + 2^{j+1}\sum_{q=q'}^{\infty}2^{q}e^{-2^{q}}2^{2^{q }}e^b B_{2^{-j}}
	= b2^{j+7}B_{2^{-j}} + 2^{j+1}B_{2^{-j}}\sum_{q=q'}^{\infty}2^{q}e^{b-2^{q}}2^{2^{q }}\enspace.
	\end{align*}

	It remains to show that $\sum_{q=q'}^{\infty}2^{q} e^{b-2^q}2^{b2^{q}}=O(1)$:
	
	\begin{align*}
	\sum_{q=q'}^{\infty}2^{q} e^{b-2^q}2^{2^{q}} 
	&=
	\sum_{r=3}^{\infty}2^{r+\log b} e^{b-2^{r+\log b}}2^{2^{{r+\log b}}} 
	=
	be^b\sum_{r=3}^{\infty}2^{r} e^{-b2^{r}}2^{b2^{r}} 
	\le
	be^b\sum_{r=3}^{\infty}2^{r} 1.3^{-b2^{r}}\enspace.
	\end{align*}
	
	For $b\ge 2$, this is decreasing in $b$ and is maximized at $b=2$:
	\begin{align*}
	\sum_{q=q'}^{\infty}2^{q} e^{b-2^q}2^{2^{q}}
	&\le2e^2\sum_{r=3}^{\infty}2^{r} 1.3^{-2^{r+1}}\enspace.
	\end{align*}
	
	For $r\ge 3$ we have $1.3^{2^{r+1}}\ge 2^{2r}$. So, 
	
	\begin{align*}\sum_{q=q'}^{\infty}2^{q} e^{b-2^q}2^{2^{q}}
	\le 2e^2\sum_{r=3}^{\infty}2^{r} 2^{-2r}\le e^2/2 < 3\enspace,\end{align*} 
	
	and we therefore reach our final bound: \[S_{\beta} \le b2^{j+7}B_{2^{-j}} + 3 \cdot 2^{j+1}B_{2^{-j}} = O(b2^{j})\enspace.\]
\end{proof}

Next, we show that there are many $j$ for which the condition of Lemma \ref{lem:cl:goodj} holds. This is because a value of $j$ for which the condition fails corresponds to a radius around $v$ about which the number of \MIS\ nodes increases rapidly. Since we have an upper bound ($\alpha$) on the size of \MIS\ over the whole graph, we know that this cannot happen at too many radii.

\begin{lemma}
	\label{cl:goodj-cond}
	The number of integers $j$, $0.01\log D \le j \le 0.1\log D$, for which there is some $r \ge 8$ satisfying $s_{j+\log b+r}> 2^{b2^{r-1}}s_{j+\log b}$ is upper bounded by $0.02\log D$.
\end{lemma}

\begin{proof}
	Consider the following process: take integers $j$ with $0.01\log D \le j \le 0.1\log D$ in increasing order. If there is some $r\ge 8$ such that $s_{j+\log b +r}> 2^{b2^{r-1}}s_{j+\log b}$, then call all values from $j$ to $j+r-1$ `bad', and continue the process from $j+r$. In this way, all $j$ that satisfy the condition of the lemma (and possibly some that do not) will be labeled `bad'.
	
	Notice that, for any particular such $j$ and $r$, we have $\prod_{i=j}^{j+r-1}\frac{s_{j+\log b+1}}{s_{j+\log b}} > 2^{b2^{r-1}} \ge 2^{\frac{r}{8}b2^{7}} = 2^{16br}$. Therefore, denoting by $q$ the number of `bad' values of $j$  at the end of the process, we have $\prod_{\text{bad }j}\frac{s_{j+\log b+1}}{s_{j+\log b}} > 2^{16bq} $.
	
	Recalling that $s_{j}$ denotes $\sum\limits_{i=0}^{2^{j+1}}\mis_i$, we have the following facts:
	
	\begin{itemize}
		\item  $s_{j}$ is non-decreasing in $j$, i.e., for any $j$, $\frac{s_{j+1}}{s_{j}} \ge 1$;
		\item $s_0 \ge 1$, since our fixed node $v$ is either in \MIS\ or has at least one neighbor in it;
		\item $s_{\log D} \le \maxis$, since \MIS\ is at most as large as the maximum independent set.
	\end{itemize} 

Using these facts, $\prod_{\text{bad }j}\frac{s_{j+\log b+1}}{s_{j+\log b}} \le \prod_{j=0}^{\log D}\frac{s_{j+\log b+1}}{s_{j+\log b}} \le \maxis$. So,  $2^{16bq}<\maxis$, and therefore $q<\frac{\log\maxis}{16b}\le\frac{\log\maxis}{64\log_D \maxis}< 0.02\log D$. That is, at most $0.02\log D$ of the values $j$ can be `bad', and satisfy the condition.
\end{proof}

We can now prove Theorem \ref{thm:cprop}, the analog of Theorem 2.2 from \cite{CD21} that represents the main technical improvement.

\begin{proof}[Proof of Theorem \ref{thm:cprop}]
	With probability at least $1-\frac{0.02}{0.1-0.01} \ge 0.77$ over choice of $j$, for all $i\ge 8$ we have that $s_{j+\log b+r}> 2^{b2^{r-1}}s_{j+\log b}$, by Lemma \ref{cl:goodj-cond}. Then, $S_{2^{-j}} = O(b2^j)$ by Lemma \ref{lem:cl:goodj}, and applying Lemma \ref{lem:edist}, we find that the expected distance from $v$ to its cluster center is at most $O(b2^j)=O(\frac{\log_D \maxis}{\beta})$.
\end{proof}

This theorem is strictly stronger than Theorem 2.2 of \cite{CD21}, which gave an expected distance to the cluster center of $O(\frac{\log_D n}{\beta})$. This improved dependence on $\alpha$ rather than $n$ carries through the remainder of the analysis of \cite{CD21} without any other changes in the proofs, yielding the following result:

\begin{theorem}
	\label{thm:compete}
	\textsc{Compete}$(S)$ informs all nodes of the highest message in $S$ within $O(D\log_D \maxis+|S|D^{0.125} +\log^{O(1)}n)$ time-steps, with high probability.
\end{theorem}

Notice that the $O(\log^3 n)$ time-steps we added by running \textsc{ComputeMIS} are absorbed in the $\log^{O(1)}n$ term.

As in \cite{CD21}, \textsc{Compete} can then be used to perform broadcasting and leader election:

\begin{theorem}
	\label{thm:broadcasting}
	\textsc{Compete}$(\{s\})$, where $s$ is the source node, completes broadcasting in undirected graphs in $O(D\log_D \maxis+\log^{O(1)} n)$ time with high probability.
\end{theorem}

\begin{proof}
	\textsc{Compete} informs all nodes of the highest message in the message set in time $O(D\log_D \maxis+\log^{O(1)} n)$, with high probability. Since this set contains only the source message, broadcasting is completed.
\end{proof}

\begin{algorithm}[H]
	\caption{{\sc Leader Election}}
	\label{alg:LE}
	\begin{enumerate}[1)]
		\item Nodes choose to become candidates in $C$ with probability $\Theta(\frac{\log n}{n})$.
		\item Candidates uniformly randomly generate $\Theta(\log n)$-bit IDs.
		\item Perform \textsc{Compete}$(C)$.
	\end{enumerate}
\end{algorithm}	

\begin{theorem}
	\label{thm:leader-election}
	Algorithm \ref{alg:LE} completes leader election in undirected graphs within time $O(D\log_D \maxis+\log^{O(1)} n)$, with high probability.
\end{theorem}

\begin{proof}
	With high probability $|C| = \Theta(\log n)$ and all candidate IDs are unique. Conditioning on this, \textsc{Compete} informs all nodes of the highest candidate ID within time $O(D\log_D \maxis+\log^{O(1)} n)$, with high probability. Therefore leader election is completed.
\end{proof}

Since growth-bounded graphs have $\maxis = poly(D)$, we obtain improved running times for both tasks therein:

\begin{corollary}
	\label{cor:leader-election}
Broadcasting and leader election can be performed in undirected growth-bounded graphs in $O(D+\log^{O(1)} n)$ rounds, succeeding with high probability.
\end{corollary}

This result naturally implies the same running time on the subclasses discussed in Section \ref{sec:classes}: (quasi) unit disk graphs, (quasi) unit ball graphs, and undirected geometric radio networks.

\section{Finding an MIS}\label{sec:MIS}
In this section we describe the \textsc{ComputeMIS} algorithm to compute a maximal independent set. Such an algorithm already exists for unit disk graphs, due to \cite{MW04} and \cite{MW05}, and indeed the algorithm of \cite{MW05} is very fast ($O(\log^2 n)$ rounds) and also has the additional property of working under asynchronous wake-up. However, the problem is more difficult in general graphs, and we are not aware of any prior algorithms for the general graph radio network model. In this section we provide the first such algorithm, running in $O(\log^3 n)$ rounds, i.e., within a $\log n$-factor of the $\Omega(\log^2 n)$ lower bound of \cite{FFM06}. When used in our broadcasting and leader election algorithms, this is dominated by the $\log^{O(1)} n$ term in those bounds.

\subsection{MIS Algorithm}

To achieve a running time of $O(\log^3 n)$ for MIS, we employ a radio network implementation of Ghaffari's algorithm for the LOCAL model \cite{Ghaffari16}. Ghaffari's algorithm was devised to improve the distributed complexity of MIS from the $O(\log n)$ of Luby's classic algorithm \cite{Luby86} to $O(\log \Delta + \log^{O(1)}\log n)$ (once equipped with the subsequent improved network decomposition of Rozho{\v{n}} and Ghaffari \cite{RG20}). Here, we use it only as an $O(\log n)$-round algorithm, simulating each step with $O(\log^2 n)$ time-steps of the radio network model. The reason for choosing Ghaffari's algorithm over Luby's is that it proves more amenable to adaptation for radio networks; while Luby's algorithm is simpler, its rounds require communication that is difficult to see how to implement in $O(\log^2 n)$ radio network time-steps\footnote{Specifically, in the standard version of Luby's algorithm, each node generates a random variable in $[0,1]$ and sends this to all its neighbors, which cannot be done efficiently in radio networks. A variant in which nodes $v$ self-nominate with probability $\Theta(1/degree(v))$ and then join the MIS if no higher-degree neighbor also self-nominated is better-suited, but checking whether higher-degree neighbors self-nominate is still problematic if the overall goal is an $O(\log^3 n)$ running time.}.

Ghaffari's algorithm is as follows (Algorithm \ref{alg:GMIS}, presentation slightly modified to make our subsequent adaptation to radio networks clearer):

\begin{algorithm}\caption{Ghaffari's MIS Algorithm \cite{Ghaffari16}}\label{alg:GMIS}
Each node $v$ maintains a desire-level $p_t(v)$ for each round $t$; $p_0(v)\gets \frac 12$;\\
	\For{ $t = 0 \rightarrow O(\log n)$}{
		$v$ \emph{marks} itself with probability $p_t(v)$;\\
		\If{$v$ has marked itself and none of its neighbors have}{
			$v$ joins \MIS	;\\
			$v$ and its neighbors remove themselves from the graph;\\
		}
		Effective degree $d_{t}(v) \gets \sum_{u\in N(v)}p_{t}(u)$;\\
		Desire-level $p_{t+1}(v)\gets\begin{cases} p_{t}/2 &\text{ if } d_{t}(v) \ge 2\\ \min\{2p_{t}(v),\frac 12\}&\text{ if } d_{t}(v) < 2\end{cases}$\\

	}
	
\end{algorithm}

There are three points at which communication is needed, which we must simulate in the radio network model: checking if any neighbors have marked themselves,  informing neighbors of joining \MIS, and determining effective degree. The first two of these we can easily accomplish using the classic \textsc{Decay} protocol, first introduced by Bar-Yehuda, Goldreich, and Itai \cite{BGI92}, for ensuring single-hop transmission success:

\begin{algorithm}[h]
	\caption{\textsc{Decay} (at a node $v$)}
	\label{alg:decay}
	\For {$i = 1$ {\textrm\textbf{to}} $\log n$}{
		$v$ transmits its message with probability $2^{-i}$.
	}
\end{algorithm}

It is a well-known property of \textsc{Decay} (see \cite{BGI92}) that if performed by a set $S$ of nodes, any node with a neighbor in $S$ hears a transmission with $\Omega(1)$ probability. By iterating $O(\log n)$ times (with sufficiently large constant within the $O()$ notation), this probability can be amplified to $1-n^{-c}$ for any constant $c$.

\begin{claim}\label{clm:decay}
	$O(\log n)$ iterations of \textsc{Decay}, performed by nodes in some set $S$, informs all nodes with a neighbor in $S$ of a message, with high probability.\qed
\end{claim}

It remains to find a way of determining effective degree in the radio network model. We first note that Ghaffari's algorithm does not require the exact $d_t(v)$ values, only whether they are less than or at least $2$ (and this threshold could be any fixed constant). It in fact suffices to give a procedure that determines (with high probability) whether $d_t(v)$ is above a constant threshold, or whether it is below a (different) constant threshold. We give an \textsc{EstimateEffectiveDegree} procedure to do this (Algorithm \ref{alg:EED}).

	\begin{algorithm}\caption{\textsc{EstimateEffectiveDegree} (at a node $v$)}\label{alg:EED}
	\For{ $i = 0 \rightarrow \log n$}{
		\For{ $C\log n$ time-steps (where $C$ is a sufficiently large constant)}{
			Node $v$ transmits with probability $\frac{p_{t}(v)}{2^i}$;\\
		}	
	}	
	If there is an $i$ for which $v$ heard at least $C\log n/33$ transmissions, return $\textsc{High}$;\\
	Otherwise, return $\textsc{Low}$;	
\end{algorithm}

The property we require is the following:

\begin{lemma}\label{lem:EED}
	If $d_t(v) \ge 1$, then \textsc{EstimateEffectiveDegree} returns $\textsc{High}$ for $v$ with high probability. If  $d_t(v) \le 0.01$, then \textsc{EstimateEffectiveDegree} returns $\textsc{Low}$ for $v$ with high probability. (If $d_t(v) \in (0.01,1)$ then either output is permitted.)
\end{lemma}

\begin{proof}

	Define $q_i(v)$ to be the probability that $v$ hears a transmission in any particular one of the time-steps of round $i$. We have:
	
	\begin{align*}
		q_i(v)
		&=\Prob{\bigcup_{u \in N(v)}u\text{ transmits and no other node in $N(v)\cup \{v\}$ does}}\\
		&= \sum_{u \in N(v)}\Prob{u\text{ transmits and no other node in $N(v)\cup \{v\}$ does}}\\
		&=  \sum_{u \in N(v)}\frac{p_{t}(u)}{2^i}  \prod_{w\in N(v)\cup \{v\}\setminus \{u\}}\left( 1-\frac{p_{t}(w)}{2^i}\right)\\
		&=  \sum_{u \in N(v)}\frac{p_{t}(u)}{2^i} \left( 1-\frac{p_{t}(v)}{2^i}\right)\left( 1-\frac{p_{t}(u)}{2^i}\right)^{-1} \prod_{w\in N(v)}\left( 1-\frac{p_{t}(w)}{2^i}\right)\enspace.
	\end{align*}

	We first show a lower bound for $q_i(v)$:
	
	\begin{align*}
		q_i(v)
		&=  \sum_{u \in N(v)}\frac{p_{t}(u)}{2^i} \left( 1-\frac{p_{t}(v)}{2^i}\right)\left( 1-\frac{p_{t}(u)}{2^i}\right)^{-1} \prod_{w\in N(v)}\left( 1-\frac{p_{t}(w)}{2^i}\right)\\
		&\ge 	\sum_{u \in N(v)}\frac{p_{t}(u)}{2^i} \left( 1-\frac{1}{2^{i+1}}\right) \prod_{w\in N(v)}\left( 4^{-\frac{p_{t}(w)}{2^i}} \right) \text{\hspace{1cm}using $1-x\ge 4^{-x}$ for $x\in [0,\frac 12]$}\\
		&\ge \frac 12	\sum_{u \in N(v)}\frac{p_{t}(u)}{2^i} 4^{-\sum_{w\in N(v)}\frac{p_{t}(w)}{2^i}}&\\
		&= 2^{-1-i} d_t(v)	4^{-2^{-i}d_t(v)}\enspace.\\
	\end{align*}
	
	To prove the first point of the lemma, when $d_t(v) \ge 1$, consider $j=\lfloor \log d_t(v)\rfloor$ (i.e., in $[0,\log n]$). Then, 
	
	\begin{align*}
		q_j(v)
		&\ge
		2^{-1-j} d_t(v)	4^{-2^{-j}d_t(v)}
		\ge 
		\frac{1}{2 d_t(v)} d_t(v)	4^{-2} = \frac{1}{32}\enspace.
	\end{align*}
	
	By linearity of expectation, the expected number of transmissions $v$ hears in round $j$ is at least $C\log n/32$. Since each time-step in the round is independent, by applying a Chernoff bound and setting $C$ sufficiently high, $v$ hears at least $C\log n/33$ transmissions in round $j$ with probability at least $1-n^{-3}$.
	
	Next, we upper-bound $q_i(v)$, in order to prove the second point of the lemma:
	
	\begin{align*}
		q_i(v)
		&=  \sum_{u \in N(v)}\frac{p_{t}(u)}{2^i} \left( 1-\frac{p_{t}(v)}{2^i}\right)\left( 1-\frac{p_{t}(u)}{2^i}\right)^{-1} \prod_{w\in N(v)}\left( 1-\frac{p_{t}(w)}{2^i}\right)\\
		&\le 	\sum_{u \in N(v)}\frac{p_{t}(u)}{2^i}\left( 1-\frac{1}{2^{i+1}}\right)^{-1} 
		\le 2	\sum_{u \in N(v)}\frac{p_{t}(u)}{2^i} 
		= 2^{1-i} d_t(v)\enspace.\\
	\end{align*}
	
	When $d_t(v) \le 0.01$, we consider any integer $i \in [0,\log n]$. Then, $q_i(v)\le 2^{1-i} 0.01 \le 0.02$.

	Therefore, the expected number of transmissions $v$ hears in any round is at most $0.02C\log n$, and again choosing $C$ to be sufficiently large, by a Chernoff bound it hears fewer than $C\log n/33$ in all rounds with probability at least $1-n^{-3}$.
	
\end{proof}

Using this \textsc{EstimateEffectiveDegree} procedure, we can now give a radio network adaptation of Ghaffari's MIS algorithm:

\begin{algorithm}\caption{Radio MIS}\label{alg:RMIS}
	Each node $v$ maintains a desire-level $p_t(v)$ for each round $t$; $p_0(v)\gets \frac 12$;\\
	\For{ $t = 0 \rightarrow O(\log n)$}{
		$v$ \emph{marks} itself with probability $p_t(v)$;\\
		Marked nodes perform $O(\log n)$ iterations of \textsc{Decay} ($O(\log^2 n)$ time-steps);\\
		\If{$v$ has marked itself and none of its neighbors have}{
			$v$ joins \MIS	;\\
		}
		\MIS\ nodes perform $O(\log n)$ iterations of \textsc{Decay};\\
		\MIS\ nodes and their neighbors remove themselves from the graph;\\
		Effective degree $d^*_{t}(v) \gets \textsc{EstimateEffectiveDegree}$;\\
		Desire-level $p_{t+1}(v)\gets\begin{cases} p_{t}(v)/2 &\text{ if } d^*_{t}(v) = \textsc{High} \\ \min\{2p_{t}(v),\frac 12\}&\text{ if } d^*_{t}(v)= \textsc{Low}\end{cases}$\\
	}
	
\end{algorithm}

\subsection{Analysis of Radio MIS}
We first note that with high probability, all of the following events occur:
\begin{itemize}
	\item Each instance of marked nodes performing \textsc{Decay} informs all nodes whether they have (at least one) marked neighbor.
	\item Each instance of \MIS\ nodes performing \textsc{Decay} informs all nodes whether they have (at least one) neighbor in \MIS.
	\item Every call to \textsc{EstimateEffectiveDegree} satisfies the properties of Lemma \ref{lem:EED}.
\end{itemize}

Consequently, we can condition on the above and analyze the algorithm assuming they hold. The analysis then follows the line of Ghaffari, with changes resulting from our differing effective degree thresholds and the fact that we allow $O(\log n)$ rounds (rather than $O(\log \Delta)$ as in \cite{Ghaffari16}).

We will call a node $v$ low-degree if $d_t(v) < 1$, and high-degree otherwise. As in \cite{Ghaffari16}, we define two types of \emph{golden round} for a node $v$, which we will later show give it a constant probability of being removed from the graph: 

\begin{itemize}
	\item Type 1: rounds in which $d_t(v) < 1$ and $p_t(v) = \frac 12$;
	\item Type 2: rounds in which $d_t(v) \ge \frac{1}{200}$ and at least $d_t(v)/10$ of it is contributed by low-degree neighbors.
\end{itemize} 

We can show that either $v$ leaves the graph, or it has $\Omega(\log n)$ golden rounds:

\begin{lemma}\label{lem:manygolden}
For any sufficiently large constant $c$, by the end of round $13c \log n$, one of the following must have occured:

\begin{itemize}
\item $v$ has joined \MIS;
\item a neighbor of $v$ has joined \MIS;
\item $v$ experiences at least $c\log n$ type-1 golden rounds;
\item $v$ experiences at least $c\log n$ type-2 golden rounds.
\end{itemize}
\end{lemma}

\begin{proof}
We consider the first $13c \log n$ rounds, and let $g_1$ and $g_2$ denote the number of golden rounds of types $1$ and $2$ for $v$ respectively. We will assume that the first three events do not happen (i.e., node $v$ is not removed and $g_1 \le c\log n$), and show that in this case, the fourth must (i.e., $g_2 > c\log n$).

Let $h$ be the number of rounds during which $d_t(v)\ge 0.01$. Note that $h$ is an upper bound on the number of rounds in which \textsc{EstimateEffectiveDegree} returns $\textsc{High}$ for $v$ and $p_t(v)$ decreases by a $2$-factor. Since the number of $2$-factor increases in $p_t(v)$ can be at most equal to the number of $2$-factor decreases, the former is also at most $h$. So, in at least $13c\log n - 2h$ of the first $13c\log n$ rounds, we have $p_t(v) = 1/2$. Out of these rounds, at most $h$ of them have $d_t(v) \ge 0.01$. As we have assumed $g_1 \le c\log n$, we get that $13c\log n - 3h \le c\log n$, i.e. $h\ge \frac13 (13c-c)\log n = 4c\log n$.

Let us consider how the effective-degree $d_t(v)$ of $v$ changes over time. If some round $t$ is not a golden round of type 2, but $d_t(v) \ge \frac{1}{200}$, then we have
\[d_{t+1}(v) \le 
2\frac{1}{10}d_t(v)+ \frac12 \frac{9}{10}d_t(v) < \frac 23 d_t(v).\]

The total effect of all $g_2$ golden rounds and all other rounds in which $d_t(v) \ge 0.01$ is therefore to reduce $d_t(v)$ by multiplying by a factor of at most $2^{g_2} (\frac 23)^{h-g_2}$. Notice that once all type-2 golden rounds have passed and $d_t(v)$ drops below $0.01$, it can never again exceed $0.01$ (since this would require a round where $d_t(v)\ge \frac{1}{200}$ and $d_{t+1}(v) \ge d_t(v)$, which can only occur in a type-2 golden round). So, we must have $2^{g_2} (\frac 23)^{h-g_2} \ge \frac{1}{200}/\frac n2$, since the combined effect of these rounds cannot have reduced $d_t(v)$ from its starting value (at most $n/2$) to below $\frac{1}{200}$. This gives 

\[
h \le \log_{2/3} (100n\cdot 2^{-g_2}) + g_2 \le 2 \log(100n\cdot 2^{g_2}) + g_2 = 3g_2 + 2\log(100 n) \enspace.
\]
 
Since, from below, $h\ge 4c\log n$, we get $g_2 >\frac13(4c\log n-2\log(100n)) > c\log n$ (so long as $c$ is sufficiently large).
\end{proof}

\begin{lemma}\label{lem:goldengood}
In any golden round for $v$, $v$ is removed from the graph with at least $1/8004$ probability. Therefore, after $13c\log n$ rounds (for sufficiently large constant $c$), no nodes remain, with high probability.
\end{lemma}

\begin{proof}
In each type-1 golden round, node $v$ gets marked with probability $1/2$. The probability that no neighbor of $v$ is marked is at least
\[\prod_{u\in N(v)}(1-p_t(u)) \ge 4^{-\sum_{u\in N(v)}p_t(u)} = 4^{-d_t(v)} > 4^{-1} = \frac 14\enspace.\] 
Here, we use that $1-x\ge 4^{-x}$ for all $x\in [0,\frac 12]$. Hence, $v$ joins the MIS with probability at least $\frac 18$.

To analyze a type-2 golden round, consider the process of revealing whether each low-degree neighbor of $v$ chooses to mark itself, in any arbitrary order, until we find one that does. Denoting by $L$ the set of low-degree neighbors of $v$, the probability that we find some marked node $u\in L$ is at least
\[1-\prod_{u\in L} (1-p_t(u)) \ge 1-e^{-\sum_{u\in L} p_t(u) } \ge 1-e^{-d_t(v)/10}\ge 1-e^{\frac{-1}{2000}}> \frac{1}{2001}\enspace.\]

Here, we use that $1-x\le e^{-x}$ for all real $x$. In this case, the probability that no neighbor of $u$ marks itself (note that while we may have revealed the choices of some of $u$'s neighbors already, they must necessarily have chosen not to mark themselves) is at least 
\[\prod_{w\in N(u)}(1-p_t(w)) \ge 4^{-\sum_{w\in N(u)}p_t(w)} = 4^{-d_t(v)} > 4^{-1} = \frac 14\enspace.\]

So, $v$ has a neighbor join \MIS\ and therefore leaves the graph with probability at least $1/8004$. We now have that in either type of golden round, $v$ leaves the graph with probability at least $1/8004$, and  by Lemma \ref{lem:manygolden} there are at least $c\log n$ golden rounds. So, if $c$ is chosen to be a sufficiently large constant, then the probability that $v$ does not get removed in the first $13c\log n$ rounds is at most $(1 - 1/8004)^{c\log n}\le n^{-3}$. Taking a union bound over all nodes $v$, the remaining graph is empty with probability at least $1-n^{-2}$.
\end{proof}

\begin{theorem}
	Algorithm \ref{alg:RMIS} computes a maximal independent set in general-graph radio networks in $O(\log^3 n)$ time-steps, succeeding with high probability.
\end{theorem}

\begin{proof}
	The algorithm clearly takes $O(\log^3 n)$ time-steps, since each round has $O(\log^2 n)$ time-steps of communication. Conditioning on the high-probability events described in Claim \ref{clm:decay}, Lemma \ref{lem:EED}, and Lemma \ref{lem:goldengood} (which affects the overall success probability by at most $n^{-2}$), we must now show that the output \MIS\ is indeed maximal and independent. Maximality follows from Lemma \ref{lem:goldengood}, since nodes are only removed if they or a neighbor joins \MIS, and so if all nodes are removed, \MIS\ is maximal. Independence clearly holds so long as all calls to \textsc{Decay} succeed: a node $v$ can only join \MIS\ in an round if none of its neighbors joined in previous rounds (since otherwise $v$ would have been removed) and none of its neighbors marks itself in \emph{this} round (and therefore, none of its neighbors joins \MIS). So, the output is a correct maximal independent set. 
\end{proof}

\section{Conclusions and Open Questions}
We have extended the broadcast and leader election algorithm of Czumaj and Davies \cite{CD21} to achieve a running time parametrized by independence number. While this algorithm works on general graphs (in $O(D\log_D\maxis+\log^{O(1)} n)$) and requires no geographic information, it yields a running time of $O(D+\log^{O(1)} n)$ in many previously-studied geometric-based graph classes. This running time has an optimal $O(D)$ leading term, and improves over previous results in such graph classes, which are either slower (for most parameter regimes) or require collision detection or other extra capabilities. We have not tried to optimize the $\log^{O(1)} n$ term, and it remains open to bring this closer to the $\Omega(\log^2 n)$ lower bound \cite{ABLP91}.

As part of this algorithm, we also provide the first algorithm for maximal independent set in general-graph radio networks. This algorithm runs in $O(\log^3 n)$ time-steps, within a $\log n$ factor of the $\Omega(\log^2 n)$ lower bound. Again, it is open whether a faster MIS algorithm is possible, or whether MIS in general graphs is in fact harder than in growth-bounded ones.

Perhaps the most interesting remaining complexity gap, though, is in the leading (diameter-dependent) term in general graphs. Our upper bound here is $O(D\log_D\maxis)$, but the only lower bound (when spontaneous transmissions are permitted) is the trivial $\Omega(D)$ bound. Any better lower bound would be interesting, and it remains to be seen whether independence number, or some other measure, is truly the `right' parameterization.

	\newcommand{\Proc}{Proceedings of the\xspace}
\newcommand{\STOC}{Annual ACM Symposium on Theory of Computing (STOC)}
\newcommand{\FOCS}{IEEE Symposium on Foundations of Computer Science (FOCS)}
\newcommand{\SODA}{Annual ACM-SIAM Symposium on Discrete Algorithms (SODA)}
\newcommand{\COCOON}{Annual International Computing Combinatorics Conference (COCOON)}
\newcommand{\DISC}{International Symposium on Distributed Computing (DISC)}
\newcommand{\ESA}{Annual European Symposium on Algorithms (ESA)}
\newcommand{\ICALP}{Annual International Colloquium on Automata, Languages and Programming (ICALP)}
\newcommand{\IPL}{Information Processing Letters}
\newcommand{\JACM}{Journal of the ACM}
\newcommand{\JALGORITHMS}{Journal of Algorithms}
\newcommand{\JCSS}{Journal of Computer and System Sciences}
\newcommand{\PODC}{Annual ACM Symposium on Principles of Distributed Computing (PODC)}
\newcommand{\SICOMP}{SIAM Journal on Computing}
\newcommand{\SPAA}{Annual ACM Symposium on Parallelism in Algorithms and Architectures (SPAA)}
\newcommand{\STACS}{Annual Symposium on Theoretical Aspects of Computer Science (STACS)}
\newcommand{\TALG}{ACM Transactions on Algorithms}
\newcommand{\TCS}{Theoretical Computer Science}
\bibliography{dbc}

\appendix
\section{Other Components of the Algorithm of \cite{CD21}}\label{app:components}

Algorithm \ref{alg:C} is conducted concurrently (via time multiplexing) with a background process (Algorithm \ref{alg:C-B}). The purpose of this process is to pass messages across the coarse cluster boundaries in the main process. It does not itself use coarse clusterings or random choices of $\beta$, and is therefore slower - broadcasting using only the background process would require $O(D\log n + \log^{O(1)}n)$ rounds. However, the analysis considers primarily the main process, and only considers the background process when messages approach coarse cluster boundaries, which occurs rarely enough to fit within the overall fast running time.

\begin{algorithm}[H]
	\caption{\textsc{Compete$(S)$ - Background Process}}
	\label{alg:C-B}
	\begin{enumerate}[1)]
		\item Compute $D^{0.2}$ different \emph{fine clusterings} using $\textsc{Partition}(\beta)$ with $\beta = D^{-0.1}$.
		\item Compute a schedule within each cluster, for each clustering.
		\item Cycling through clusterings in round-robin order, perform \textsc{Intra-Cluster Propagation$(O(\frac{\log n}{\beta}))$}.
	\end{enumerate}
\end{algorithm}

To achieve fast communication within clusters, \textsc{Intra-Cluster Propagation} uses fast schedules from \cite{GHK13} as implemented in \cite{HW16}:

\begin{algorithm}[H]
	\caption{\textsc{Intra-Cluster Propagation$(\ell)$}}
	\label{alg:ICP}
	\begin{enumerate}[1)]
		\item Broadcast the highest message known by the cluster center to all nodes within distance $\ell$.
		\item All such nodes which know a higher message participate in a broadcast towards the cluster center.
		\item Broadcast the highest message known by the cluster center to all nodes within distance $\ell$.
	\end{enumerate}
\end{algorithm}

\textsc{Intra-Cluster Propagation} also has a background process, conducted concurrently via time multiplexing as before. The purpose of this process is to work around collisions caused by nodes bordering other fine clusters, which can interrupt the fast schedules.

\begin{algorithm}[H]
	\caption{\textsc{Intra-Cluster Propagation Background Process}}
	\label{alg:ICP-B}
	Repeat until main process is complete:\\
	\For {$i = 1$ {\textrm\textbf{to}} $\log n$}{
		with probability $2^{-i}$ (coordinated in each cluster) perform one iteration of \textsc{Decay};\\
		otherwise remain silent for $\log n$ steps.
	}
\end{algorithm}	

This background process makes use of the classic \textsc{Decay} protocol by Bar-Yehuda, Goldreich, and Itai \cite{BGI91}, discussed here as Algorithm \ref{alg:decay}.

\end{document}